\documentclass[12pt]{article}

\usepackage{color}
\usepackage{complexity,fullpage}
\usepackage{amsmath,amssymb,amsthm}
\usepackage{graphicx,bbm,calc,capt-of,ifthen}
\usepackage[utf8]{inputenc}
\usepackage{graphicx}
\usepackage[normalem]{ulem}
\usepackage{algorithmic,algorithm}
\usepackage{setspace}

\newcommand{\Clique}{\textrm{\bf CLIQUE}}
\newcommand{\PMatch}{\textrm{\bf PMATCH}}

\newcommand{\size}{\textrm{\bf Size}}
\newcommand{\depth}{\text{\bf Depth}}
\newcommand{\depthp}{\text{\bf Depth}^{+}}

\newcommand{\KW}{\text{\bf KW}}
\newcommand{\KWP}{\text{\bf KW}^{\mathbf{+}}}

\newcommand{\set}[1]{\left\{ #1 \right\}}

\newtheorem{theorem}{Theorem}
\newtheorem{claim}[theorem]{Claim}
\newtheorem{proposition}[theorem]{Proposition}
\newtheorem{definition}{Definition}
\newtheorem{corollary}{Corollary}
\newtheorem{lemma}{Lemma}
\newtheorem{remark}{Remark}
\newenvironment{proposition-a}[1]{\noindent {\bf Proposition~#1:~}\em }{\smallskip}
\newenvironment{lemma-a}[1]{\noindent {\bf Lemma~#1:~}\em }{\smallskip}
\newenvironment{theorem-a}[1]{\noindent {\bf Theorem~#1:~}\em }{\smallskip}

\usepackage{color}

\title{Depth Lower Bounds against Circuits with Sparse Orientation}
\author{Sajin Koroth\thanks{Department of CSE, Indian Institute of
    Technology Madras, Chennai 600036, India.} \and Jayalal Sarma$^{*}$}
\date{\today}
\begin{document}

\maketitle

\begin{abstract} 
  We study depth lower bounds against non-monotone circuits,
  parametrized by a new measure of non-monotonicity: the
  orientation\footnote{A generalization of monotone functions are
    studied under the name {\em unate functions}(cf.~\cite{IPS97}). We
    inherit the terminology of {\em orientation} from that setting. We
    remark that our definition is universal unlike the case of unate
    functions.} of a function $f$ is the characteristic vector of the
  minimum sized set of negated variables needed in any
  DeMorgan\footnote{Circuits where negations appear only at the
    leaves.}  circuit computing $f$. We prove trade-off results
  between the depth and the weight/structure of the orientation
  vectors in any circuit $C$ computing the $\Clique$ function on an
  $n$ vertex graph. 
  We prove that if $C$ is of depth $d$ and each gate computes a
  Boolean function with orientation of weight at most $w$ (in terms of
  the inputs to $C$), then $d \times w$ must be $\Omega(n)$. In
  particular, if the weights are $o(\frac{n}{\log^k n})$, then $C$
  must be of depth $\omega(\log^k n)$.  We prove a barrier for our
  general technique. However, using specific properties of the
  $\Clique$ function (used in \cite{amano2005superpolynomial}) and the
  Karchmer-Wigderson framework~\cite{Karchmer:1988:MCC:62212.62265},
  we go beyond the limitations and obtain lower bounds when the weight
  restrictions are less stringent. We then study the depth lower
  bounds when the structure of the orientation vector is
  restricted. Asymptotic improvements to our results (in the
  restricted setting), separates $\NP$ from $\NC$.  As our main tool,
  we generalize Karchmer-Wigderson
  game~\cite{Karchmer:1988:MCC:62212.62265} for monotone functions to
  work for non-monotone circuits parametrized by the weight/structure
  of the orientation. We also prove structural results about
  orientation and prove connections between number of negations and
  weight of orientations required to compute a function.
\end{abstract}

\section{Introduction}
Deriving size/depth lower bounds for Boolean circuits computing
$\NP$-complete problems has been one of the main goals of circuit
complexity.
By a counting
argument\cite{riordan1942number} it is known that ``almost'' all
Boolean functions require exponential size and linear depth. 
Despite many efforts the best known lower bound against an explicit
function computed by general circuits is still a constant multiple of
the number of inputs~\cite{IwamaKazuoMorizumiHiroki5nBound}.
And the best known depth lower bound for an explicit function against
general circuits of bounded fan-in is (derived
from formula size lower bound due to H\aa
stad~\cite{Hastad98theshrinkage}) less than $3 \log n$.
Attempts to prove size lower bounds against constant depth
circuits has yielded useful results (see survey~\cite{All96,All08} and
textbook~\cite{juknatext}).


Notable progress has been made in proving lower bounds against
monotone circuits. Monotone circuits are circuits without negtions
gates. Such circuits can only compute monotone functions. Montone
functions are Boolean functions whose value does not decrease when
input bits are changed from $0$ to
$1$. Razborov~\cite{razborov1968lower} proved a super-polynomial size
lower bound against monotone circuits computing the $\Clique$ function
which is $\NP$-hard. This was further strengthened to an exponential
lower bound by Alon and Boppana~\cite{AB87}. A super polynomial lower
bound is known~\cite{razborovlogicalpermanent} also against monotone
circuits computing $\PMatch$ problem.  Since $\PMatch$ is known to be
in $\P$~\cite{Edmonds1965a} it has polynomial size circuits. Thus it
shows that non-monotonocity is helpful in reducing size even when the
function computed is monotone.

Moving in the direction of non-monotonicity, Amano and
Maruoka~\cite{amano2005superpolynomial} established a super-polynomial
lower bound against circuits computing the $\Clique$ function with at
most $\frac{1}{6}\log \log n$ negations. A chasm was already known at
the $\log n$ negations; Fisher~\cite{FischerLimitedNegation} proved
that any circuit of polynomial size can be converted to a circuit of
polynomial size that has only $\log n$ negations. In particular, this
implies that if we are able to extend the technique of lower bounds to
work against circuits having $\log n$ negations, then it separates
$\P$ from $\NP$.
Jukna~\cite{Jukna200471} further tightened the gap for explicit
multi-output functions by establishing a
super-polynomial size lower bound against circuits with at most $\log
n - 16 \log \log n$ negations.

It is known~\cite{RazMatchingSuperLinearDepth} that both $\Clique$
function and $\PMatch$ function on $n$-vertex graphs require
$\Omega(n)$ depth when computed by bounded fan-in monotone circuits.
Thus, non-monotonicity is useful
in the depth restricted setting also, as $\PMatch$ is known to be in
non-uniform $\NC^2$~\cite{Lov79}.  A main technique involved in the
monotone depth lower bound for
$\PMatch$~\cite{RazMatchingSuperLinearDepth} is a characterization of
circuit depth using a communication game defined between two
players. Raz and Wigderson~\cite{Raz89probabilisticcommunication} used
this framework to obtain a lower bound of $\Omega(n^2)$ on the number
of negations at the leaves for any $O(\log n)$ depth DeMorgan circuit
solving the $s$-$t$ connectivity problem.  However, we do not
know\footnote{Indeed, size lower bounds against bounded fan-in
  circuits in the presence of
  negations~\cite{amano2005superpolynomial} also imply depth lower
  bounds against them. In particular, \cite{amano2005superpolynomial}
  implies that any circuit with $\frac{1}{6}\log \log n$ negation
  gates computing $\Clique(n,(\log n)^{\sqrt{\log n}})$ requires depth
  $\Omega((\log n)^{\sqrt{\log n}})$.} any depth lower bound which uses
Karchmer Wigderson framework against
circuits where there are negations at arbitrary locations.

\subsection{Our Results}

We study an alternative way of limiting the non-monotonicity in the
circuit. To arrive at our restriction, we define a new measure called
{\em orientation} of a Boolean function. 

\begin{definition}
A Boolean function $f : \set{0,1}^n
\to \set{0,1}$ is said to have {\em orientation} $\beta \in
\set{0,1}^n$ if there is a monotone Boolean function $h: \set{0,1}^{2n} \to
\set{0,1}$ such that :
$\forall x \in \set{0,1}^n, f(x) = h ( x,(x \oplus \beta) )$.
\end{definition}

If $f$ is a monotone Boolean function, from the above defintion it is
clear that all-$0$'s vector is an orientation of $f$.  The {\em weight
  of an orientation} is simply the number of $1$'s in $\beta$, and can
be thought of as a parameter indicating how ``close'' $f$ is to a
monotone function.
Note that for any DeMorgan circuit computing $f$, the characteristic
vector of negated input indices form an orientation vector of the
function $f$.  Because replacing the negated input variables
with fresh variables results of a DeMorgan circuit results in a
monotone circuit.


The definition can be extended to circuits as well. We consider
circuits where the function computed at each gate can be non-monotone.
But each gate computes a function whose orientation (with respect to
the in inputs of the circuits) must be of limited weight.
We say a circuit $C$ is weight $w$
oriented if every internal gate of $C$ computes a function which has
an orientation $\beta$ with $|\beta| \leq w$. 
The weight restriction on a circuit thus defined is a semantic
restriction as we are only limiting the weight of orientation of
the functions computed at sub-circuits of $C$. But we do not place any
restriction on how (especially interms of actual negation gates) the
functions at sub-circuits are computed in $C$.
 We prove the following theorem which
presents a depth vs weight trade-off for weight restricted circuits.

\begin{theorem}
\label{introthm:mainlb}
If $C$ is a Boolean circuit of depth $d$ and weight of the orientation
$w$ ($w > 0$), computing $\Clique$ then, $d \times w$ must be
$\Omega(n)$.
\end{theorem}

In particular, if the weights are $o(\frac{n}{\log^k n})$, the
$\Clique$ function requires $\omega(\log^k n)$ depth. By contrast, any
circuit computing $\Clique$ has weight of the orientation at each gate
at most $n^2$. We prove the above theorem by extending the
Karchmer-Wigderson framework for monotone circuit depth to the case of
non-monotone circuits which are ``sparsely oriented''.  The proof
depends critcally on the route to monotone depth via
Karchmer-Wigderson games. This is because it is unclear how to
directly simulate weight $w$-restricted circuit model using a monotone
circuit for $w >0$. We remark that the above theorem applies even to
circuits computing $\PMatch$.

The difficulty in extending the above lower bound to more general
lower bounds is the potential presence of gates computing ``densely''
oriented functions. In this context, we explore the usefulness of
gates with non-zero orientation in a circuit. We argue that allowing
even a constant number of non-zero (but ``dense'') weight of
orientation gates can make the circuit more powerful in the limited
depth setting.  In particular, we show (see Theorem~\ref{thm:2dense})
that:
\begin{theorem}
  There exists a monotone Boolean function $f$ which cannot be
  computed by poly-log depth monotone circuits, but there is a
  poly-log depth circuit computing $f$ such that there are at most two
  internal gates whose weight of orientation is non-zero. 
\end{theorem}

We note that the function in Theorem~\ref{thm:2dense} is derived as a
restriction from the non-uniform $\NC^2$ circuit computing $\PMatch$
and hence is {\em not explicit}.  The above theorem indicates that the
densely oriented gates are indeed useful, and that
Theorem~\ref{introthm:mainlb} cannot be improved in terms of the
number of densely oriented gates it can handle, without using specific
properties about the function (for example, $\Clique$) being computed.

Going beyond the above limitation, we exploit the known properties of
the $\Clique$ function and use the generalized Karchmer-Wigderson
games to prove lower bounds against circuits with less stringent
weight restrictions (in particular, we can restrict the weight
restrictions to only negation gates and their inputs)
\begin{theorem}
\label{introthm:am}
For any circuit family $\mathcal{C}=\set{C_m}$ (where $m = {n \choose
  2}$) computing $\Clique(n,n^{\frac{1}{6\alpha}})$ with $\ell+k$
negation gates, where $\ell \leq 1/6 \log \log n$,
$\alpha=2^{\ell+1}-1$, at least $k$ negation gates are computing
functions which are sensitive only on $w$ inputs\footnote{i.e., the
  weight of orientation of the function computed at their input plus
  orientation of the function computed at their output is at most $w$}
with $kw \leq \frac{n}{8}$
and the remaining $\ell$ negations compute functions of arbitrary
orientation: $\depth(C_m) \geq n^{\frac{1}{ 2^{\ell+8}}}$
\end{theorem}
This theorem implies that $\Clique$ cannot be computed by circuits
with depth $n^{o(1)}$ even if we allow some constant number of gates
to have non-zero (even dense) orientation - thus going beyond the
earlier hurdle presented for $\PMatch$. We remark that the above
theorem also generalizes the case of circuits with negations at the leaves
($\ell = 0$, and $w=1$).

It gives hope that by using
properties of $\Clique$ (like hardness of
approximation~\cite{AmanoMaruokaPotentialOfApproximation} used by
~\cite{amano2005superpolynomial}) one can possibly push the technique
further.

We also explore the question of the number of
densely oriented gates that are required in an optimal depth
circuit. We establish the following connection to the number of
negations in the circuit.

\begin{theorem}\label{thm:negations-to-orientation}
  For any circuit $C$ with $t$ negations, there is a circuit $C'$
  computing the same function such that $\size(C') \leq 2^t \times
  (\size(C) + 2^t) + 2^t$, and there are at most $2^{t-1} (t + 2) - 1$
  internal gates whose orientation is a non-zero vector.
\end{theorem}

Next we study circuits where the structure of the orientation is
restricted.
The restriction is on the number of vertices of the input
graph involved in edges indexed by the orientation vector of the
function.

\begin{theorem}\label{thm:vertex-index-beta}
  If $C$ is a circuit computing the $\Clique$ function and for each
  gate $g$ of $C$, the number of vertices of the input graph involved
  in edges indexed by $\beta_g$ (the orientation vector of gate $g$)
  is at most $w$, then $d \times w$ must be $\Omega(\frac{n}{\log
    n})$.
\end{theorem}

We also study a sub-class of the above circuits for which we prove
better lower bounds.  A circuit is said to be of \textit{uniform
  orientation} if there exists a single orientation vector $\beta \in
\set{0,1}^n$ such that every gate in it computes a function for which
$\beta$ is an orientation vector.

\begin{theorem}\label{thm:uniform-beta-lb}
  Let $C$ be a circuit computing the $\Clique$ function, with uniform
  orientation $\beta \in \{0,1\}^n$ such that there
  is a subset of vertices $U$, $|U| \ge \log^{k+\epsilon} n$ for which
  $\beta_e = 0$ for all edges $e$ within $U$, then $C$ must have depth
  $\omega(\log^k n)$.
\end{theorem}

We remark that a DeMorgan circuit has an orientation of weight exactly
equal to the number of negated variables. However, this result is
incomparable with that of \cite{Raz89probabilisticcommunication}
against DeMorgan circuits for two reasons : (1)~this is for the
$\Clique$ function. (2)~the lower bounds and the class of circuits are
different.

In contrast to the above theorem, we show that an arbitrary circuit
can be transformed into one having our structural restriction on the
orientation with $|U| = O(\log^k n)$.
\begin{theorem}
\label{thm:betareduction}
If there is a circuit $C$ computing $\Clique$ with depth $d$ then for
any set of $c\log^k n$ vertices $U$, there is an equivalent circuit
$C^{'}$ of depth $d+ c \log^k n$ with orientation $\beta$ such that none
of the edges $e(u,v),\, u,v \in U$ has $\beta_{e(u,v)}=1$.
\end{theorem}

Thus if either Theorem~\ref{thm:uniform-beta-lb} is extended to $|U| =
\Omega(\log^k n)$ or the transformation in Theorem~\ref{thm:betareduction}
can be modified to give $|U| = O(\log^{k+\epsilon} n)$ for some
constant $\epsilon > 0$, then a depth lower bound for $\Clique$
function against general circuits of depth $O(\log^k n)$ will be
implied.


\section{Preliminaries}

For $x,y \in \set{0,1}^n$, $x \leq y$ if and only if for all $i \in
[n]$, $x_i \leq y_i$. A Boolean function $f$ is said to be monotone if
for all $x \leq y$, $f(x) \leq f(y)$. In other words value of a
monotone function does not decrease when input bits are changed from
$0$ to $1$.

For a set $U$, we denote by $\binom{U}{2}$ the set $\set{\set{u,v} |
  u,v \in U}$.  In an undirected graph $G=(V,E)$, a clique is a set $S
\subseteq V$ such that $\binom{S}{2} \subseteq E(G)$.  $\Clique(n,k)$
is a Boolean function $f : \set{0,1}^{\binom{n}{2}} \to \set{0,1}$
such that for any $x \in \set{0,1}^{\binom{n}{2}}$, $f(x)=1$ if $G_x$,
the undirected graph represented by the undirected adjacency matrix
$x$ has a clique of size $k$.  $\Clique(n,k)$ is a monotone function
as adding edges (equivalent to turning $0$ to $1$ in adjacency matrix)
cannot remove a $k$-clique, if one already exists. By $\Clique$, we
denote $\Clique(n,\frac{n}{2})$.  A perfect matching of an undirected
graph $G=(V,E)$ is a $M \subseteq E(G)$ such that no two edges in $M$
share an end vertex and it is such that every vertex $v \in V$ is
contained as an end vertex of some edge in $M$.  Corresponding Boolean
function $\PMatch : \set{0,1}^{\binom{n}{2}} \to \set{0,1}$ is defined
as $\PMatch(x)=1$ if $G_x$ contains a perfect matching.  Note that
$\PMatch$ is also a monotone function.

A circuit is a directed acyclic graph whose internal nodes are labeled
with $\land$, $\lor$ and $\lnot$ gates, and leaf nodes are labeled
with inputs. The function computed by the circuit is the function
computed by a designated ``root'' node. All our circuits are of
bounded fan-in.  The depth of a circuit $C$, denoted by $\depth(C)$ is
the length of the longest path from root to any leaf, and $\depth(f)$
denotes the minimum possible depth of a circuit computing $f$. By
$\depth_{t}(f)$ we denote the minimum possible depth of a circuit
computing $f$ with at most $t$ negations. Size of a circuit is simply
the number of internal gates in the circuit, and is denoted by
$\size(C)$. $\size(f),\size_{t}(f)$ are defined analogous to
$\depth(f),\depth_{t}(f)$ respectively.  We refer the reader to a
standard textbook~(cf. \cite{vollmertext}) for more details.
%

We now review the Karchmer-Wigderson games and the related lower bound
framework. The technique is a strong connection between circuit depth
and communication complexity of a specific two player game where the
players say Alice and Bob are given inputs $x \in f^{-1}(1)$ and $y
\in f^{-1}(0)$, respectively. In the case of general circuits, the
game is denoted by $\KW(f)$ and the goal is to find an index $i$ such
that $x_i \ne y_i$. In the case of monotone circuits, the game is
denoted by $\KWP(f)$ and the goal is to find an index $i$ such that
$x_i = 1$ and $y_i = 0$. Since monotone circuits compute monotone
functions $\KWP(f)$ defined only for monotone Boolean functions
$f$. We abuse the notation and use $\KW(f)$ and $\KWP(f)$ to denote
the number of bits exchanged in the worst case for the best protocol
solving the corresponding communication game. Karchmer and
Wigderson~\cite{Karchmer:1988:MCC:62212.62265} proved that for any
function $f$ best possible depth any circuit computing $f$, denoted by
$\depth(f)$ is equal to $\KW(f)$. And for any monotone function $f$ the
best possible depth of any monotone circuit computing $f$, denoted by
$\depthp(f)$ is equal to $\KWP(f)$. Raz and
Wigderson~\cite{RazMatchingSuperLinearDepth} showed that
$\KWP(\Clique)$ and $\KWP(\PMatch)$ are both $\Omega(n)$.

\subsection{Characterization of Orientation}
We recall the definition of orientation :
 a function $f : \set{0,1}^n \to \set{0,1}$ is said to have {\em
    orientation} $\beta \in \set{0,1}^n$ if there is a monotone
  function $h: \set{0,1}^{2n} \to \set{0,1}$ such that :
    $\forall x \in \set{0,1}^n, f(x) = h ( x,(x \oplus \beta) )$.
Thus, if $\beta \in \set{0,1}^n$ is an orientation for a
function $f$, then any $\beta' \geq \beta$ is also an orientation for
$f$ by definition of orientation.
This
is because one can interpret orientation vector $\beta$ as an advice
containing negation of a subset variables so that $f$ can be written as a
monotone function of the input along with the negated variables in
$\beta$. Hence for any superset of these negated variables ($\beta' >
\beta$) it would still be possible to write $f$ as a monotone function
of the input and negated variables in $\beta$ alone.

We first show that any function $f(x)$ can be written in the form of
definition as an $h(x,x\oplus \beta)$ for a monotone function $h$.
Let $C$ be any circuit computing $f$. Convert $C$ into a DeMorgan
circuit $C'$ by pushing down the negations via repeated applications of
De-Morgan's law. In $C'$ replace every $\bar{x_i}$ with a new
variable $y_i$ for every $i\in [n]$. Thus $C'$ on inputs $x,y$ is a
monotone function. Since there are $n$ input variables at most $n$
$y_i$'s are needed. Let $h=C'(x_1,\dots,x_n,y_1,\dots,y_n)$ be the
monotone function computed by $C'$ after replacing the negated inputs
by fresh variables. Clearly $h$ satisfies the required form with
$\beta$ defined as $\beta_i=1$ if and only if $\bar{x_i}$ appears in
$C'$. Now if the function $f$ has such a form, take any monotone
circuit $C_h$ computing $h$. Replace all the inputs $x_i \oplus
\beta_i$ where $\beta_i=1$ with $\bar{x_i}$ and all the inputs $x_i
\oplus \beta_i$ where $\beta_i=0$ with $x_i$ in $C_h$. Thus we get a
circuit $C''$ computing $f$, which is De-Morgan and has negations only
on variables where $\beta_i=1$. Thus for any function $f$ whose
orientation is $\beta$, there is a circuit $C$ of uniform orientation
$\beta$. This is because a sub-circuit rooted at any gate of $C''$ is also
a De-Morgan circuit and has negated variables which are a subset of negated
variables in $C''$.

We now establish that orientation is a well-defined measure.  We
prove a sufficient condition for the $\beta_i$ to be $1$ in any
orientation for a function $f$.

\begin{proposition}\label{thm:negatedimpliesone}
  For any function $f$, if there exists a pair $(u,v)$ such that
  $u_i=0,v_i=1$, $u_{[n]\setminus \set{i}}=v_{[n]\setminus \set{i}}$
  and $f(u)=1,f(v)=0$ then any orientation $\beta$ of the function
  must have $\beta_i=1$.
\end{proposition}
\begin{proof}
  Let $h$ be the monotone function corresponding to $f$ for $\beta$
  such that $\forall x, f(x)=h(x,x\oplus \beta)$.  Assume to the
  contrary that $\beta_i=0$. Since $u_{[n]\setminus \set{i}}=v_{[n]
    \setminus \set{i}}$, we have that $u_{[n]\setminus \set{i}} \oplus
  \beta =v_{[n] \setminus \set{i}} \oplus \beta$ for any
  $\beta$. Hence $(u,u\oplus \beta),(v,v\oplus \beta)$ differs only in
  two indices, namely $i,n+i$. At $i$, $u_i=0,v_i=1$, and at $n+i$
  since $\beta_i=0$, $u_{n+i}=0,v_{n+i}=1$. Hence we get that
  $(u,u\oplus \beta) < (v,v\oplus \beta)$, but $h(u,u\oplus
  \beta)=1,h(v,v\oplus \beta)=0$ a contradiction to monotonicity of
  $h$. 
\end{proof}

It is not a priori clear that the minimal (with respect to $<$
relation on the Boolean hypercube $\set{0,1}^n$) orientation for a function
$f$ is unique. We prove that it is indeed unique.

\begin{proposition}
  \label{prop:uniquebeta}
  Minimal orientation for a function $f : \set{0,1}^n \to \set{0,1}$ is
  well defined and it is $\beta \in \set{0,1}^n$ such that $\beta_i=1$
  if and only if there exists a pair $(u,v)$ such that $u_i=0,v_i=1$,
  $u_{[n]\setminus \set{i}}=v_{[n]\setminus \set{i}}$ and
  $f(u)=1,f(v)=0$.
\end{proposition}

\begin{proof}
  From Proposition~\ref{thm:negatedimpliesone} it is clear that any
  orientation $\beta'$ of a function $f$ is such that $\beta \leq
  \beta'$.  We claim that negations of variables in $\beta$ suffices
  to compute $f$ using a DeMorgan circuit. Define a partial function
  $h: \set{0,1}^{2n} \to \set{0,1}$ associated with orientation
  $\beta$ of $f$ as $h(x,x\oplus \beta) \triangleq f(x)$.  We claim
  that this partial function has an extension which is a monotone
  function. We claim that for any $u,v\in \set{0,1}^n$ such that
  $u\leq v$ and $f(u)=1,f(v)=0$, there exists an $i\in [n]$ such that
  $u_i =0,v_i=1$ and $\beta_i=1$. Let $w_0=u \leq w_1 \leq \cdots \leq
  w_j \leq w_{j+1} \leq \cdots \leq w_k=v$ be a chain between $u$ and
  $v$. Take the minimum $j$ such that $f(w_j)=1$ and
  $f(w_{j+1})=0$. Since $w_j,w_{j+1}$ satisfies assumptions of
  Proposition~\ref{thm:negatedimpliesone}, for the $i$ where $w_j$ and
  $w_{j+1}$ differs, $\beta_i=1$. Since $u \leq w_j$ and $i$th bit of
  $w_j$ is $0$, we get $u_i=0$. Similarly $v_i=1$ as $v \geq w_{j+1}$
  and $i$th bit of $w_{j+1}$ is $1$. With this claim we can prove that
  for any $(s,s \oplus \beta)$ and $(t,t \oplus \beta)$ either they
  are incomparable or $f(s) \geq f(t)$ if and only if $(s,s \oplus
  \beta) \geq (t,t \oplus \beta)$. Assume to the contrary that $f(s) <
  f(t)$ and $(s,s \oplus \beta) \geq (t,t \oplus \beta)$. Since $(s,s
  \oplus \beta) \geq (t,t \oplus \beta)$, $s \geq t$ and
  $f(s)=0,f(t)=1$ as $f(s) < f(t)$. But then we are guaranteed by the
  earlier claim an $i \in [n]$ such that
  $s_i=1,t_i=0,\beta_i=1$. Since $\beta_i=1$, $s_i \oplus \beta_i=0$
  and $t_i \oplus \beta_i=1$ whereas $s_i=1,t_i=0$ implying that $(s,s
  \oplus \beta) \not \geq (t,t \oplus \beta)$, a contradiction. Thus
  the partial function we defined will never have a chain with a $1$
  to $0$ transition. Also any partial function $h$ which does not have
  a $1 \to 0$ transition on any of the chains of the Boolean
  hypercube, has an extension to a function which is monotone.
\end{proof}


\section{Lower Bound Argument for Sparsely Oriented Circuits}

In this section, we prove Theorem~\ref{introthm:mainlb} which shows
the trade-off between depth and weight of orientation of the internal
gates of a circuit. We prove the following lemma which is the main
contribution of our paper.

\begin{lemma}
  If $C$ is a depth $d$ circuit computing a monotone Boolean function $f :
  \set{0,1}^n \to \set{0,1}$ which is sensitive on all its inputs
  and each internal gate of $C$ computes a Boolean function whose
  orientation has weight at most $w$, then $d \times (4w+1) \ge
  \KWP(f)$.
\end{lemma}

\begin{proof} The proof idea is to devise a protocol for $\KWP(f)$
  using $C$ having $\depth(C)$ rounds and each round having a
  communication cost of $4w+1$.

  Alice is given $x \in f^{-1}(1)$ and Bob is given $y \in
  f^{-1}(0)$. The goal is to find an index $i$ such that
  $x_i=1,y_i=0$. The protocol is described in
  Algorithm~\ref{algo:protocolKWModified}.

\begin{spacing}{0.8}
\begin{algorithm}[h]
\caption{Modified Karchmer-Wigderson Protocol}
\label{algo:protocolKWModified}
\begin{algorithmic}[1]
  \STATE\COMMENT{Let $x'$ and $y'$ be the current inputs. At the
    current gate $g$ computing $f$, with the input gates $g_1$ and
    $g_2$ where $f_1$ and $f_2$ are the functions computed, let
    $\beta_1,\beta_2$ be the corresponding orientations (and are known
    to both Alice and Bob).  If $g_1$ or $g_2$ is a negation gate, let
    $\gamma_1$ and $\gamma_2$ be the orientation vectors of input
    functions to $g_1$ and $g_2$, otherwise they are $0$-vectors. Let
    $\alpha = \beta_1 \lor \beta_2 \lor \gamma_1 \lor \gamma_2$.  Let
    $S = \{i: \alpha_i = 1\}$, $x_S$ is the substring of $x$ indexed
    by $S$.  } \IF{$g$ is $\land$} \STATE{Alice sends $x'_S$ to
    Bob. Bob compares $x'_S$ with $y'_S$.}  \IF{there is an index $i
    \in S$ such that $x'_i=1$ and $y'_i=0$}\label{step:indexi}
  \STATE{Output $i$}.  \ELSE
  \STATE{Define $y'' \in \set{0,1}^n$:  $y''_S = x'_S$ and $y''_{[n]\setminus S} = y'_{[n]\setminus S}$.} \\
  \STATE{Bob sends $i \in \{1,2\}$ such that $f_i(y'') = 0$ to
    Alice. They recursively run the protocol on $g_i$ with $x'=x'$ and
    $y'=y''$.}\label{step-not-stuck-and}
      \ENDIF
\ENDIF
\IF{$g$ is $\lor$} 
      \STATE{Bob sends $y'_S$ to Alice. Alice compares $y'_S$ with $x'_S$.}
      \IF{there is an index $i \in S$ such that $x'_i=1$ and $y'_i=0$}
      \STATE{Output $i$}.
      \ELSE
      \STATE{Define $x'' \in \set{0,1}^n$:  $x''_S = y'_S$ and $x''_{[n] \setminus S} = x'_{[n] \setminus S}$.} \\
	  \STATE{Alice sends $i \in \{1,2\}$ such that $f_i(x'') = 1$ to Bob. They recursively run the protocol on $g_i$ with $x'=x''$ and $y'=y'$.}\label{step-not-stuck-or}
      \ENDIF
\ENDIF
\end{algorithmic}
\end{algorithm}
\end{spacing}

\vspace{3mm}
We now prove that the
protocol (Algorithm~\ref{algo:protocolKWModified}) solves
$\KWP(f)$. The following invariant is maintained during the run
of the protocol and is crucial for the proof of correctness.

\textbf{Invariant:} When the protocol is at a node which computes a
function $f$ with orientation vector $\beta$ it is guaranteed a priori
that the inputs held by Alice and Bob, $x'$ and $y'$ are equal on the
indices where $\beta_i=1$, $f(x')=1,f(y')=0$ and restriction of $f$
obtained by fixing variables where $\beta_i=1$ to $x'_i (=y'_i)$ is a
monotone function.

If the invariant is maintained, we claim that when the
protocol stops at a leaf node of the circuit computing a function
$f$ with $f(x')=1$ and $f(y')=0$ then $f=x_i$ for some $i\in [n]$. If
the leaf node is a negative literal, say $\bar{x_i}$ then by
Proposition~\ref{thm:negatedimpliesone}, orientation of $\bar{x_i}$
has $\beta_i=1$. By the guarantee that $x'_\beta=y'_\beta$,
$x'_i=y'_i$, contradicting $f(x')\neq f(y')$. Hence whenever the
protocol stops at a leaf node it is guaranteed that the leaf is labeled
by a positive literal.  And when input node is labeled by a positive
literal $x_i$, then a valid solution is output as $f(x')=1,f(y')=0$
implies $x'_i=1$ and $y'_i=0$.  Note that during the run of the
protocol we only changed $x,y$ at indices $i$ where  $x_i \neq y_i$, to
$x'_i = y'_i$. Hence, any index where $x'_i \neq y'_i$ it is the case
that $x_i=x'_i$ and $y_i=y'_i$.
  
Now to prove the invariant note that it is vacuously true at the root
gate as $f$ is a monotone function implying $\beta=0^n$, and in the
standard $\KWP(f)$ game $x \in f^{-1}(1)$ and $y \in f^{-1}(0)$. We
argue that, while descending down to one of the children of the
current node the invariant is maintained. To begin with, we show that
the protocol does not get stuck in step~\ref{step-not-stuck-and} (and
similarly for step~\ref{step-not-stuck-or}). To prove this, we claim
that at an $\land$ gate $f=f_1 \land f_2$, if the protocol failed to
find an $i$ in step~\ref{step:indexi} such that $x'_i=1,y'_i=0$ then
on the modified input $y''$ at least one of $f_1(y'')$ or $f_2(y'')$ is
guaranteed to be zero.  Since the protocol failed to output an $i$
such that $x'_i=1,y'_i=0$, it must be the case that $x'_i \leq
y'_i$ for indices indexed by $\beta_1,\beta_2$.  Let $U$ be the subset
of indices indexed by $\beta_1$ and $\beta_2$ where $x_i=0$ and $y_i=1$.
Bob obtains $y''$ from $y'$ by setting $y''_i=0$ for all $i\in U$.
Thus we have made sure that $x'$ and $y''$ are the same on the variables
whose negations are required to compute $f,f_1$ and $f_2$.

Consider the functions $f',f'': \set{0,1}^{n-|\beta_1 \lor \beta_2|}
\to \set{0,1}$ which are obtained by restricting the variables indexed
by orientation vectors of $f_1$ and $f_2$ to the value of those
variables in $x'$.  Both $f'$ and $f''$ are monotone as they are
obtained by restricting all negated input variables of the DeMorgan
circuits computing $f_1$ and $f_2$ for orientations $\beta_1$ and
$\beta_2$ respectively.  The changes made to $x',y'$ were only at
places where they differed.  Thus at all the indices where $x',y'$
were same, $x',y''$ is also same.  Hence monotone restriction
$f_{x'_\beta}$ of $f$ obtained by setting variables indexed by $\beta$
to their values in $x'$ is a consistent restriction for $y''$ also.
Note that $y'' \leq y'$. Hence $f(y'')=0$ because $y''$ agrees with
$y'$ on variables indexed by $\beta$ (as $x''$ agrees with $y'$ and
$y''$ on variables indexed by $\beta$) implying
$f_{x'_\beta}(y''_{[n]\setminus \beta})\leq
f_{x'_\beta}(y'_{[n]\setminus \beta})=0$.  Since $f(y'')=0$, it is
guaranteed that one of $f_1(y''),f_2(y'')$ is equal to $0$.  Bob sets
$y'=y''$ and sends $0$ if it is $f_1(y'')=0$ or $1$ otherwise,
indicating Alice which node to descend to. Note that
$x'_{\beta_1}=y''_{\beta_1}$, $x'_{\beta_2}=y''_{\beta_2}$ and
restriction of $f_1,f_2$ to $x'_{\beta_1},x'_{\beta_2}$ respectively
gives monotone functions $f',f''$ thus maintaining the invariant for
both $f_1$ and $f_2$.
  
We claim that if any of the input gates $g_1,g_2$ to the current
$\land$ gate $g$ is a $\neg$ gate then the protocol will not take the
path through the negation gate.  To argue this, we use the following
lemma.  
\begin{lemma}\label{lemma:NegationSensitivity}
  If $\ell,\bar{\ell}$ are functions with orientations
  $\beta,\gamma$, 
  then for all $x,y \in \set{0,1}^n$
  such that $x_{\beta \lor \gamma}=y_{\beta \lor \gamma}$, $\ell(x)=\ell(y)$.
\end{lemma}
\begin{proof}
  We know that for a function $\ell$, if there exists a pair $(u,v) \in
  \set{0,1}^n \times \set{0,1}^n$ with $u \leq v$, $u_i \neq v_i$,
  $u_{[n]\setminus \set{i}}=v_{[n]\setminus \set{i}}$ and $\ell(u)=1,
  \ell(v)=0$ then by Proposition \ref{thm:negatedimpliesone} for every
  orientation $\beta$, $\beta_i=1$ . Let $i$ be an index on which
  $\ell$ is sensitive, i.e., there exists $(u,v) \in \set{0,1}^n \times
  \set{0,1}^n$ with $u \leq v$, $u_i \neq v_i$, $u_{[n]\setminus
    \set{i}}=v_{[n]\setminus \set{i}}$ and $\ell(u) \neq \ell(v)$. Note
  that $l$ is sensitive on $i$ need not force $\beta_i=1$, as it
  could be that $\ell(u)=0$ and $\ell(v)=1$. But in this case
  $\bar{\ell}(u)=1$ and $\bar{\ell}(v)=0$, hence $\gamma_i = 1$ for
  $\bar{\ell}$. Hence, $\ell$ is sensitive only on indices in $\beta
  \lor \gamma$.
  \qed
\end{proof}
The lemma establishes that every negation gate in a weight $w$
oriented circuit computes a function which is sensitive on at most
$2w$ indeces.
 Hence if $2w < n$ the root gate cannot be a negation gate
for a function sensitive on all inputs.  Suppose at a node one of the
children is a negation gate, say $f_1$.  Since we ensure $x'_{\beta_1
  \lor \gamma_1}=y''_{\beta_1 \lor \gamma_1}$, Lemma~\ref{lemma:NegationSensitivity} implies
$f_1(x')=f_1(y'')$. But the protocol does not descend down a path
where $x',y''$ are not separated.  Hence the claim.
  
This also proves that when the protocol reaches an $\land$ node where
both children are negation gates, at the round for that node protocol
outputs an index $i$ and stops.  Otherwise, since we ensure
$x'_{S}=y''_{S}$, $f_1(y'')=f_1(x')=1$ and $f_2(y'')=f_2(x')=1$ by
Lemma~\ref{lemma:NegationSensitivity}.  But this contradicts the fact
that at a node $f = f_1 \land f_2$ either $f_1(y'')=0$ or $f_2(y'')=0$
(or both).

Proof of equivalent claims for an $\lor$ gate is similar except for
the fact that Alice modifies her input.  

Thus, using the above protocol we are guaranteed to solve $\KWP(f)$.
Communication cost of any round is at most $4w+1$.  Because if any of
the children is a negation gate then we have to send its orientation
along with the orientation of its complement. The protocol clearly
stops after $\depth(C)$ many rounds. Thus communication complexity of
the protocol is upper bounded by $\depth(C) \times (4w+1)$.
\end{proof}


\section{Dense Orientation}

Currently our depth lower bound technique cannot handle orientations
of weight $\frac{n}{\log^k n}$ or more for obtaining $\omega(\log^k
n)$ lower bounds. In light of this, we explore the usefulness of
densely oriented gates in a circuit. First we prove that any
polynomial sized circuit can be transformed into an equivalent circuit
of polynomial size but having only $O(n \log n)$ gates of non-zero
orientation by studying the connection between orientations and
negations. Next we present a limitation of our technique in a circuit
having only two gates of non-zero (but ``dense'') orientation. Thus,
strengthening of our technique will have to use some property of the
function being computed. Finally we show how to use a property of
$\Clique$ function to slightly get around the limitation.

\subsection{From Negation Gates to Orientation}

Since weight of the orientation can be thought of as a measure of
non-monotonicity in a circuit, a natural question to explore is the
connection between the number of negations and number of non-zero
orientations required to compute a function $f$. We show the following: 
\begin{theorem}
  For any function $f : \set{0,1}^n \to \set{0,1}$, if there is a
  circuit family $\set{C_n}$ computing $f$ with $t(n)$ negations then
  there is also a circuit family $\set{C'_n}$ computing $f$ such
  that $\size(C'_n) \leq 2^t \times (\size(C_n) + 2^t) + 2^t$, and
  there are at most $2^{t-1} (t + 2) - 1$ internal gates whose
  orientation is non-zero.
\end{theorem}

\begin{proof}
  In $C_n$ replace input of each negation by new a variable, say
  $y_1,\dots,y_t$, thus obtaining a circuit
  $C^{''}_n(x_1,\dots,x_n,y_1,\dots,y_t)$.  Let $g_1,\dots,g_t$ be the
  inputs to the $t$ negation gates (in topologically sorted order) of
  $C_n$.  Note that for each setting of $y_1,\dots, y_t$ to some $b
  \in \set{0,1}^t$, $C^{''}_n(x,b)$ is monotone circuit computing a
  monotone function on $x_1,\dots,x_n$. Hence the orientation of each
  internal gate in $C^{''}_n(x,b)$ is zero.  Let $g_{i,b}$ for $i \in
  [t], b \in \set{0,1}^t$ denote the monotone function computed by the
  sub-circuit $C_{g_i}$ of $C_n$ rooted at gate $g_i$, where
  $g_1,\dots,g_{i-1}$ are set to $b_1,\dots,b_{i-1}$ respectively.
  Thus we can write f as:
  \begin{eqnarray}\label{eqn:NegationToOrientation}
    f(x_1,\dots,x_n) =  \bigvee_{b \in \set{0,1}^t} \left(
      \bigwedge_{i=1}^{t} g_{i,b}^{b_i}(x) \right) C^{''}_n(x,b), 
  \end{eqnarray}
  where $g^{0}$ denotes $\overline{g}$ and $g^{1}$ denotes $g$.  When
  $t=1$, then the above expression becomes
  $f(x)=g(x)C(x,1)+\overline{g}(x)C(x,0)$.  In this case the only
  gates which can have non-zero orientation are the negation
  computing $\overline{g}$, $\land$ computing $\overline{g}(x)C(x,0)$
  and the root gate (if the function computed is non-monotone).
Hence when $t=1$ the circuit has at most three gates
  with non-zero orientation if the circuit computes a  non-monotone
  function and at most two gates of non-zero orientation otherwise.

  Consider the formulation of a circuit $C'$ computing $f$ given in
  Equation~\ref{eqn:NegationToOrientation}. Clearly $\size(C') \leq
  2^t \times (\size(C_n) + 2^t) + 2^t$. All internal gates in
  $C''_n(a,b)$ are monotone. Non-zero orientation is needed only for
  computing:
  \begin{itemize}
  \item $\underset{i\in [t],b_i=0}{\bigwedge} \overline{g_{i,b}}$ 
  \item  $\land$ of   
    $\underset{i\in [t],b_i=0}{\bigwedge} \overline{g_{i,b}}$ with
    $\underset{i\in [t],b_i=1}{\bigwedge} g_{i,b} \land C^{''}(x,b)$
  \item the $\lor$-tree, computing $\sum$ of $2^t$ terms which are
    potentially of non-zero orientation.
  \end{itemize}
  For computing $\underset{i\in [t],b_i=0}{\bigwedge}
  \overline{g_{i,b}}$, we need an $\land$ tree of $t-|b|_1$ many
  leaves. Number of internal nodes in the tree is $t-|b|_1-1$ (for
  $t>1$). To compute the $\land$ of this intermediate product with
  $\underset{i\in [t],b_i=1}{\bigwedge} g_{i,b} \land C^{''}(x,b)$ one
  more gate is need. Thus the total number of gates needed is
  $t-|b|_1$. Let us call number of such gates $K_1$. By the above
  analysis, $K_1 = \sum_{b \in \set{0,1}^t} (t-|b|_1) = t \times
  2^{t-1}$.  The remaining gates are the internal gates in the $\lor$
  tree implementing the sum of terms. Since there are $2^t$ leaves,
  number of internal nodes in the tree, say $K_2$ is $2^t-1$. Hence
  total number of nodes with non-zero orientation is at most $K_1 +
  K_2 = 2^{t-1} (t + 2) - 1$.
\end{proof}

\begin{remark}
In conjunction with the result of
Fisher~\cite{FischerLimitedNegation}, this implies that it is enough
to prove lower bounds against circuits with at most $O(n \log n)$
internal nodes of dense orientations, to obtain lower bounds against
the general circuits.
\end{remark}

\subsection{Power of Dense Orientation}
\label{sec:powerdenseorientation}

We show that even as few as two ``densely'' oriented internal gates can
help to reduce the depth from super poly-log to poly-log for some
functions.

\begin{theorem}
\label{thm:2dense}
There exists a monotone Boolean function $f$ such that it cannot be
computed by poly-log depth monotone circuits, but there is a poly-log
depth circuit computing it with at most two internal gates have
non-zero orientation $\beta$.
\end{theorem}

\begin{proof}
  It is known~\cite{RazMatchingSuperLinearDepth} that $\PMatch$ does
  not have monotone circuits of poly-log depth. But if arbitrary
  negations are allowed then there is an $O(\log^2 n)$ depth circuit
  computing $\PMatch$~\cite{Lov79}. Monotone function $f$ claimed in
  the theorem is obtained from poly-log depth circuit $C$ computing
  $\PMatch$. Fischer's theorem guarantees that without loss of
  generality we can assume that $C$ has at most
  $\log n$ negations.

  If there is a poly-log depth circuit having exactly one negation
  computing $\PMatch$,
  then Theorem~\ref{thm:negations-to-orientation} can be applied to
  get a circuit of poly-log depth having at most two gates of
  non-zero orientation.  Otherwise, the circuit has $t \geq 2$ negations,
  and there is no poly-log depth circuit computing the same function
  with one negation.  Let $g_1$ denote the input to the first negation
  gate(in the topological sorted order) in $C$.  From $C$ obtain $C'$
  by replacing $g_1$ with a new variable, say $y_1$.  Let $C'_0$,
  $C'_1$ denote the circuits obtained by setting $y_1$ to $0$,$1$
  respectively.  The corresponding functions $f_0$, $f_1$ need not be
  monotone. Hence we define monotone functions $f'_0,f'_1$ from $f_1,
  f_0$ :
  \begin{eqnarray*}
    f'_0 \left( x \right) &=& f_0 \left( x \right) \vee g_1 \left( x
    \right) \\
    f'_1 \left( x \right) &=& f_1 \left( x \right) \wedge g_1 \left( x \right)
  \end{eqnarray*}
  When $g_1 \left( x \right) = 0$, $f_0 \left( x \right) = f \left( x
  \right)$ and when $g_1 \left( x \right) = 1$, $f'_0 \left( x \right) =
  1$. Hence $f'_0 $ is monotone. A similar argument can be used to
  establish that $f'_1$ is monotone. Note that both $f'_0,f'_1$ have
  poly-log depth circuits computing it with at most $t-1$ negation
  gates.

  We claim that one of $f'_0, f'_1$ does not have a monotone circuit
  of poly-log depth.  Otherwise from poly-log depth monotone circuits
  computing $f_0', f_1'$ and the monotone circuit of poly-log depth
  computing $g_1 $ we can get a poly-log depth circuit computing $f$
  with one negation : use $\overline{g_1}(x)$ as a selector to select
  $f_1' \left( x \right) $ or $f_0' \left( x \right)$ as which is
  appropriate. This circuit computes $f$ because, by definition,
  $(g_1(x) \wedge f'_1 (x)) \vee ( \overline{g_1}(x) \wedge f'_0 (x))
  = f(x)$.
  This contradicts our assumption that there is no circuit
  of poly-log depth computing $f$ with one negation.

  Applying the procedure once, we get a monotone function $f'$ which
  has a $t-1$ negation poly-log depth circuit computing it, but it has
  no monotone circuit of poly-log depth computing it.  If the function
  $f'$ has a poly-log depth circuit with one negation then
  Theorem~\ref{thm:negations-to-orientation} can be applied to get the
  desired function.  Otherwise apply the procedure on $f'$ as $f'$ is
  a monotone function which does not have any poly-log depth circuit
  with at most one negation computing it.  Applying the procedure at
  most $t$ ($t \leq \log n$) times we get to a monotone function $f'$
  having a poly-log depth circuit with one negation, but has no
  monotone poly-log depth circuit computing it. Applying
  Theorem~\ref{thm:negations-to-orientation} on the one negation
  circuit gives a poly-log depth circuit with at most two gates of
  non-zero orientation.
\end{proof}

This theorem combined with the ``sparse'' orientation protocol implies
that the two non-zero orientations $\beta_1,\beta_2$ is such that
$|\beta_1|+|\beta_2|$ is not only non-zero but is super
poly-log. Because our protocol will spend $|\beta_1|+|\beta_2|$ for
handling these two gates, and on the remaining gates in the circuit it
will spend $1$ bit each. Hence the cost of the sparse orientation
protocol will be at most $|\beta_1|+|\beta_2|+\depth(C)$. Thus
$|\beta_1|+|\beta_2|$ is at least $\KWP(f) - \depth(C)$ which is super
poly-log as $\depth(C)$ is poly-log and $\KWP(f)$ is super
poly-log.

\begin{remark}
By Theorem~\ref{thm:2dense} we get a function which has an $\NC^2$
circuit with two non-zero orientation gates which has no monotone
circuit of poly-log depth. Thus our bounds cannot be strengthened to
handle higher weight 
 without incorporating the specifics of the function being
computed.  In section~\ref{sec:LbClique}, we rescue the situation
slightly using the specific properties of the $\Clique$ function.
\end{remark}

\begin{remark}
  The proof of Theorem~\ref{thm:2dense} also implies that there is a
  monotone function $f$ (not explicit) such that there is a one
  negation circuit in $\NC^2$ computing it, but any monotone circuit
  computing $f$ requires super-poly-log depth.
\end{remark}

\subsection{Lower Bounds for \Clique\ function}
\label{sec:LbClique}

The number of gates with high orientations can be arbitrary in
general. In this subsection we give a proof for
Theorem~\ref{introthm:am}.  We first extend our technique to handle
the low weight negations efficiently so that we get a circuit on high
weight negations (see
Lemma~\ref{lemma:limitednegationsarbitorientaiton} below). To complete
the proof of Theorem~\ref{introthm:am}, we appeal to depth lower
bounds against negation-limited circuits computing $\Clique(n,n^{\frac{1}{6}\alpha})$. 
\begin{lemma}
  \label{lemma:limitednegationsarbitorientaiton}
  For any circuit family $\mathcal{C}=\set{C_n}$ computing a monotone
  function $f$ where there are $k$ negations in $C_n$ computing
  functions which are sensitive only on $2w$ inputs bits (i.e., the
  orientation of their input as well as their output is at most $w$)
  with $kw \leq \frac{n}{8}$ and the remaining $\ell$ negations
  compute functions of arbitrary orientation: $\depth(C_n) \geq
  \depth_{\ell}(\Clique(\frac{3n}{4},n^{\frac{1}{6}\alpha})) $
\end{lemma}
\begin{proof}

  Since $k$ negations of $C_n$ are depended only on $kw$ inputs (i.e,
  edges) the number of vertices which has at least one of its edges
  indexed by one of the $k$ negations is $2kw$. Let this set of
  vertices be $S$ and $|S|\leq \frac{n}{4}$. In $C_n$ set input
  variables corresponding to edges in $\binom{S}{2}$ and the variables
  corresponding to edges between $S$ and $[n]\setminus S$ to $0$. Note
  that the circuit $C'_n$ obtained from $C_n$  by this restriction 
  computes $\Clique(\frac{3n}{4},n^{\frac{1}{6}\alpha})$.
  Note that all the $k$
  negations which are sensitive only on edges indexed by
  $\binom{S}{2}$ is fixed to constans as $\binom{S}{2}$ is
  fixed. Hence $C'_n$ has at most $\ell$ negations. Hence the theorem.
\end{proof}
By a straight forward application of technique used in
\cite{amano2005superpolynomial} to prove size lower bounds against
circuits with limited negations computing $\Clique(n,n^{\frac{1}{6} \alpha})$ we obtain the size version of
following lemma (For completeness, we include the relevant part in the
Appendix ~\ref{app:am}).

\begin{lemma}
  \label{lemma:depthnegtradeoff}
  For any circuit $C$ computing $\Clique(n,n^{\frac{1}{6\alpha}})$
  with $\ell$ negations where $\ell \leq 1/6 \log \log n$ and
  $\alpha=2^{\ell+1}-1$,
  \begin{eqnarray*}
    \depth_{\ell}(f) &\geq& n^{\frac{1}{81 \alpha}}
  \end{eqnarray*}
\end{lemma}

Combining Lemma~\ref{lemma:limitednegationsarbitorientaiton} and
Lemma~\ref{lemma:depthnegtradeoff} completes the proof of
Theorem~\ref{introthm:am}.


\section{Structural Restrictions on Orientation}

In this section we study structural restrictions on the orientation
and prove stronger lower bounds.

\subsection{Restricting the Vertex Set indexed by the Orientation}

We first consider restrictions on the set of vertices\footnote{Notice
  that the input variables to the \Clique\ function represents the
  edges. This makes the results of this section incomparable with the
  depth lower bounds of \cite{Raz89probabilisticcommunication}.}
indexed by the orientation - in order to prove
Theorem~\ref{thm:vertex-index-beta} stated in the introduction. As in
the other case, we argue the following lemma, which establishes the
trade-off result. By using the lower bound for $\KWP$ games for
$\Clique$ function, the theorem follows.
\begin{lemma}
  Let $C$ be a circuit of depth $d$ computing $\Clique$, with each gate computing a function whose orientation 
  is such that the number of vertices of the input graph   
  indexed by the orientation $\beta$ is at most $\frac{w}{\log n}$, then  
$d$ is $\Omega\left(\frac{\KWP(f)}{4w+1}\right)$.
\end{lemma}

\begin{proof}
  It is enough to solve the $\KWP(f)$ on the min-term, max-term pairs
  which in case of $\Clique(n,k)$ is a $k$-clique and a complete
  $k-1$-partite graph. We play the same game as in the proof of
  Theorem~\ref{introthm:mainlb}, but instead of sending edges we send
  vertices included in the edge set indexed by $\beta$ with some
  additional information. If it is Alice's turn, then $x'_{\beta}$
  defines an edge sub-graph of her clique. Both Alice and Bob know
  $\beta$ and hence knows which vertices are spanned by edges
  $e_{u,v}$ such that $\beta_{e(u,v)}=1$. So Alice can send a bit
  vector of length at most $w$ (in the case of Alice we can handle up
  to $w$), indicating which of these vertices are part of her
  clique. This information is enough for Bob to deduce whether any
  $e_{u,v}$ indexed by $\beta$ is present in Alice's graph or
  not. Since Bob makes sure that $x'_{\beta}=y'_{\beta}$ by modifying
  his input, and Alice keeps her input unchanged, Alice knows what
  modifications Bob has done to his graph.

  Similarly on Bob's turn, he sends the vertices in the partition
  induced by $y_{\beta}$ and the partition number each vertex belongs
  to (hence the $\log n$ overhead for Bob) to Alice. With this
  information Alice can deduce whether any $e_{u,v} \in \beta$ is
  present in Bob's graph or not. Inductively they maintain that they
  know of the changes made to other parties input in each round. Hence
  the game proceeds as earlier. This completes the proof of the
  theorem.
\end{proof}

\subsection{Restricting the Orientation to be Uniform}

In this section, we consider the circuits where the orientation is
uniform and study its structural restrictions. We proceed to the proof
of Theorem~\ref{thm:uniform-beta-lb}.
\begin{theorem-a}{\ref{thm:uniform-beta-lb}}
  Let $C$ be a circuit computing the $\Clique$ function with uniform
  orientation $\beta \in \{0,1\}^n$ such that there
  is subset of vertices $U$ and $\epsilon > 0$ such that $|U| \ge
  \log^{k+\epsilon} n$ for which $\beta_e = 0$ for all edges $e$
  within $U$, then $C$ must have depth $\omega(\log^k n)$.
\end{theorem-a}
\begin{proof}
  We prove by contradiction. Suppose there is a circuit $C$ of depth
  $c\log^k n$. In the argument below we assume $c=1$ for
  simplicity. Without loss of generality, we assume that $|U| =
  \log^{k+\epsilon }n$. Fix inputs to circuit $C$ in the following way:
  \begin{itemize}
  \item 
    Choose an arbitrary $K_{\frac{n}{2}-\frac{|U|}{2}}$ comprising of
    vertices from $[n] \setminus U$ and set those edges to $1$.
  \item 
    For every edge in $\binom{[n] \setminus U}{2}$ which is
    not in the clique chosen earlier, set to $0$.
  \item 
    For every edge between $[n] \setminus U$ and $U$ set
    it to $1$.
  \end{itemize}
  Since every edge $e(x,y)$ which has $\beta_e=1$ has at least one of
  the end points in $[n] \setminus U$, by above setting, all those
  edges are turned to constants. Thus we obtain a monotone circuit
  $C''$ computing $\Clique(|U|,\frac{|U|}{2})$ of depth at most $(\log
  n)^k$. In terms of the new input, $(\log n)^k=((\log
  n)^{k+\epsilon})^{\frac{k}{k+\epsilon}} =
  (|U|)^{\frac{k}{k+\epsilon}}$, this contradicts the
  Raz-Wigderson~\cite{RazMatchingSuperLinearDepth} lower bound of
  $\Omega(|U|)$, as $\frac{k}{k+\epsilon} < 1$ for $\epsilon > 0$.
  \end{proof}

  Note that for Clique function, with the above corollary we can
  handle up to weight $\frac{n^2}{(\log n)^{2+2\epsilon}}$ if the
  vertices spanned by $\beta$ is up to $\frac{n}{(\log
    n)^{1+\epsilon}}$ and still get a lower bound of $(\log
  n)^{1+\epsilon}$. This places us a little bit closer to the goal of
  handling $\beta$ of weight $n^2$, from handling just $(\log
  n)^{1+\epsilon}$.\\ {\bf A contrasting picture:} Any function has a
  circuit with a uniform orientation $\beta=1^{n}$ ($|\beta|=n$). We
  show that the weight of the orientation can be reduced at the
  expense of depth, when the circuit is computing the $\Clique$
  function.\\
  
  \begin{theorem-a}{~\ref{thm:betareduction}}
    If there is a circuit $C$ computing $\Clique$ with depth $d$ then
for any set of $c\log n$ vertices $U$, there is an equivalent circuit
$C^{'}$ of depth $d+ c \log n$ with orientation $\beta$ such that none
of the edges $e(u,v),\, u,v \in U$ has $\beta_{e(u,v)}=1$.
  \end{theorem-a}
  \begin{proof}
    The proof idea is to devise a $\KW$ protocol based on circuit $C$
    such that for $e(u,v)$ where $u,v \in U$ the protocol is
    guaranteed to output in the monotone way, i.e., $x_{e(u,v)}=1$ and
    $y_{e(u,v)}=0$. The modified protocol is as follows:
    \begin{itemize}
    \item Alice chooses an arbitrary clique $K_\frac{n}{2} \in G_x$ (which she is
      guaranteed to find as $x\in f^{-1}(1)$). She then obtains $x'$
      by deleting edges $e(x,y)$ from $\binom{U}{2}$ which are
      outside the chosen clique $K_{\frac{n}{2}}$. Note that since $K_\frac{n}{2}
      \in G_{x'}$, $f(x')=1$.
    \item Alice then sends the characteristic vector of vertices in
      $K_\frac{n}{2} \cap U$ which is of length at most $c \log
      n$ to Bob. 
    \item Bob then obtains $y'$ from $y$ by removing edges
      in $\binom{U}{2}$ which are outside the clique formed by
      $K_\frac{n}{2} \cap \binom{U}{2}$. By monotonicity of $\Clique$
      $f(y')=0$.
    \item If there is an edge $e(u,v) \in K_\frac{n}{2} \cap \binom{U}{2}$
      which is missing from $y'$ Bob outputs the index
      $e(u,v)$. Otherwise they run the standard Karchmer-Wigderson game
      on $x',y'$ using the circuit $C$ to obtain an $e(x,y)$ such
      that $e(x,y)$ is exclusive to either $G_{x'}$ or $G_{y'}$.
    \end{itemize}

    The cost of the above protocol is $d + c \log n$.  For any $e(u,v)
    \in E(G) \setminus \binom{U}{2}, x'_{e(u,v)}=x_{e(u,v)}$ and
    $y'_{e(u,v)}=y_{e(u,v)}$.  The protocol never answers
    non-monotonically($i,x'_i=0, y'_i=1$) for an edge $e(u,v)$ with
    $u,v \in U$.  Because our protocol ensures that for any $e \in
    \binom{U}{2}$ , $x'_{e} \geq y'_{e}$, ruling out such a
    possibility. By the connection between $\KW(f)$ and circuit depth,
    we get a circuit having desired properties.
  \end{proof}

Thus we get the following corollary.
\begin{corollary}
If there is a circuit $C \in \NC^{k}$ computing $\Clique(n,k)$, then
there is a circuit $C^{'} \in \NC^{k}$ of uniform orientation $\beta$
computing $\Clique(n,k)$ such that there are $(c \log n)^k$ vertices
$V^{'}$ with none of the edges $e(u,v)$
 having $\beta_{e(u,v)}=1$.
\end{corollary}
\begin{proof}
 It follows by setting $d= O((\log n)^k)$ and modifying the protocol
 to work over a $V^{'}$ of size $(c \log n)^k$. The analysis and
 proof of correctness of the protocol remains the same, but the
 communication cost becomes $O((\log n)^k)+(c \log n)^k=O((\log  n)^k)$.
\end{proof}
In other words, if we improve Theorem~\ref{thm:vertex-index-beta} to
the case when the orientation ``avoids'' a set of $\log n$ vertices
(instead of $(\log n)^{(1+\epsilon)}$ as done), it will imply $\NC^1
\ne \NP$.


\section{Discussion and Open Problems}
In this work, we studied lower bounds against non-monotone circuits with a new measure of non-monotonicity - namely the orientation of the functions computed at each gate of the circuit.
As the first step, we proved that the lower bound can be obtained by modifying the Karchmer-Wigderson game. We studied the weight of the orientation of the functions at internal gates as a parameter of the circuit, and explored the usefulness of densely oriented gates. We also showed the connections between negation limited circuits and orientation limited circuits. A main open problem that arises from our work is to improve upon the weight restriction of the orientation vector ($\Omega(\frac{n}{\log n})$) for which we can prove depth lower bounds. 

\section{Acknowledgements}
We thank the anonymous referees for the useful comments.
\bibliographystyle{plain} 
\bibliography{references}
\newpage

\appendix

\noindent {\Large \bf Appendix}

\section{Proof of Lemma~\ref{lemma:depthnegtradeoff} - Choice of
  parameters in~\cite{amano2005superpolynomial}}
\label{app:am}

In this section we give the arguments for
Lemma~\ref{lemma:depthnegtradeoff}.  Since the trade-off result stated
in Lemma~\ref{lemma:depthnegtradeoff} is not explicitly stated and
proved in~\cite{amano2005superpolynomial}, in this section, we present
the relevant part of the proof technique in
\cite{amano2005superpolynomial} with careful choice of parameters
obtaining the trade-off. For
consistency with notation used in~\cite{amano2005superpolynomial},
\textit{for the remainder of this section we will be denoting the
  number of vertices in the graph by $m$.}

The main idea in~\cite{amano2005superpolynomial} is to consider the
boundary graph of a function $f$, defined as $G_f=\set{(u,v)|
  \Delta(u,v)=1, f(u)\neq f(v)}$ where $\Delta(u,v)$ is the hamming
distance. They prove that if there is a $t$ negations circuit $C$
computing $f$ then the boundary graph $f$ must be covered by union of
boundary graphs of $2^{t+1}$ functions obtained by replacing the
negations in $C$ by variables and considering the input functions of
$t$ negation gates and the output gate where the negations in the
sub-circuit considered are restricted to constants.

They prove that,

\begin{lemma}{~\cite[Theorem
    3.2]{amano2005superpolynomial}}\label{lemma:boundarycovering}
  Let $f$ be a monotone function on $n$ variables. For any positive
  integer $t$,
  \begin{displaymath}
    \size_t(f) \geq \underset{F'=\set{f_1,\dots,f_\alpha}
      \subseteq \mathcal{M}^n}{\min}\set{\underset{f'\in F'}{\max
        \set{size_{mon}(f')}} \mid \bigcup_{f'\in F'} G(f') \supseteq G(f)}
  \end{displaymath}
  where $\alpha=2^{t+1}-1$ and $G(f')$ denotes the boundary graph of the
  function $f'$.
\end{lemma}

The size lower bound they derive crucially depends on the following
lemma which states that no circuit of ``small'' size can ``approximate''
clique in the sense that either it rejects all the ``good'' graphs or
accepts a huge fraction of ``bad'' graphs.

\begin{lemma}{~\cite[Theorem
    4.1]{amano2005superpolynomial}}
  Let $s_1,s_2$ be positive integers such that $64 \leq s_1 \leq s_2$
  and $s_1^{1/3}s_2 \leq \frac{m}{200}$. Suppose that $C$ is a monotone
  circuit and that the fraction of good graphs in $I(m,s_2)$ such that
  $C$ outputs $1$ is at least $h=h(s_2)$. Then at least one of the
  following holds:
  \begin{itemize}
  \item The number of gates in $C$ is at least $(h/2)2^{s^{1/3}/4}$.
  \item The fraction of bad graphs in $O(m,s_1)$ such that $C$ outputs
    $0$ is at most $2/s_1^{1/3}$.
  \end{itemize}
\end{lemma}

where a ``good'' graph in $I(m,s_2)$ is a clique of size $s_2$ on $m$
vertices and no other edges and a ``bad'' graph in $O(m,s_1)$ is an $(s_1 -
1)$-partite graph where except for at most one partition the
partitions are balanced and of size $\lceil \frac{m}{s_1-1} \rceil$
each. 

\begin{lemma}
  \label{lemma:constantnegationsizebound}
  For any circuit $C$ computing $\Clique(m,m^{\frac{1}{6\alpha}})$ with
  $t$ negations with $t \leq 1/6 \log \log m$, size of
  $C$ is at least $2^{m^{\frac{1}{81 \alpha}}}$ where
  $\alpha=2^{t+1}-1$.
\end{lemma}

\begin{proof}
  The proof is similar to the proof (\cite[Theorem
  5.1]{amano2005superpolynomial}) by Amano and Maruoka except for change
  of parameters. Assume to the contrary that there is a circuit $C$
  with at most $t$ negations computing
  $\Clique(m,m^{\frac{1}{6\alpha}})$ with size $M$, $M <
  2^{m^{\frac{1}{81\alpha}}}$. By Lemma~\ref{lemma:boundarycovering}
  there are $\alpha \triangleq 2^{t+1}-1$ functions
  $f_1,\dots,f_\alpha$ of size at most $M$ (as they are obtained by
  restrictions of the circuit $C$) such that
  $\cup_{i=1}^{\alpha}G(f_i) \supseteq G(f)$. Let
  $s=m^{\frac{1}{6\alpha}}$ and let $l_0,l_1,\dots,l_\alpha$ be a
  monotonically increasing sequence of integers such that
  $l_0=s,l_\alpha=m$ and $l_i=m^{1/10+(i-1)/(3\alpha)}$. Note that
  $s^{1/3}l_{i} \leq l_{i+1}$ as
  $s^{1/3}l_{i}=m^{1/(18\alpha)+1/10+(i-1)/(3\alpha)}<m=^{1/10+(i)/(3\alpha)}=l_{i+1}$. Also
  $\left[ l_0=s=m^{\frac{1}{6\alpha}} \right] < \left[ l_1= m^{1/10}
  \right]$ as $\alpha =2^{t+1}-1 \geq 2^2-1$, $l_{\alpha-1} <
  m^{1/10+1/3} < m$. Thus, $l_0 < l_1 < \dots < l_i < l_{i+1} < \dots
  < l_\alpha$. The definition of ``bad'' graphs and ``good'' graphs at
  layer $l_i$ remains the same as in
  \cite{amano2005superpolynomial}. Note that \cite[Corollary
  5.2]{amano2005superpolynomial} is true for our choice of parameters
  as $s^{1/3}l_{i-1} \leq l_i$. Equations 5.1 to 5.3 of
  \cite{amano2005superpolynomial} is valid in our case also as these
  equations does not depend on the value of the parameters. The
  definition of a dense set remains the same, and $h \geq
  \frac{1}{\alpha} \geq \frac{1}{m}$ (as $m \geq \log m \geq \alpha$)
  is such that $(h/2)2^{s^{1/3}/4} \geq \frac{1}{m} 2^{m^{\frac{1}{18
        \alpha}}/4}$ is strictly greater than $M=2^{m^{\frac{1}{81
        \alpha}}}$. Hence Equation 5.4 of
  \cite{amano2005superpolynomial} is also true in our setting. Claim
  5.3 of \cite{amano2005superpolynomial} is independent of choice of
  parameters, hence is true in our setting also.
  \begin{claim}{~\cite[Claim
      5.3]{amano2005superpolynomial}}\label{claim:badgraphnextlayer}
    
    Suppose $c_1 > 1$ and $c_2 > 1$. Put $c_3 = \alpha$. Let
    $f_1,\dots,f_{c_3}$ be the monotone functions such that
    $\cup_{i=1}^{c_3}G(f_i) \supseteq G(\Clique(m,s))$ and
    $size_{mon}(f_i) \leq M$ for any $1 \leq i \leq c_3$. Suppose that
    for distinct indices $i_1,\dots, i_k \in [c_3]$,
    \begin{displaymath}
      \Pr_{L_k \in \mathcal{L}_k}\left[ \Pr_{u\in O_{L_k}}\left[
          f_{i_1}(u)=\cdots=f_{i_k}(u)=1\right] \geq \frac{1}{c_1}\right] \geq \frac{1}{c_2}
    \end{displaymath}
    holds. If $c_1c_2c_3 \leq s_1^{1/3}/8$, then there exists $i_{k+1}
    \in [c_3] \setminus \set{i_1,\dots,i_k}$ such that 
    \begin{displaymath}
      \Pr_{L_{k+1} \in \mathcal{L}_{k+1}}\left[ \Pr_{u\in O_{L_{k+1}}}\left[
          f_{i_1}(u)=\cdots=f_{i_k}(u)=1\right] \geq \frac{1}{4c_1c_2c_3}\right] \geq \frac{1}{2c_1c_2}
    \end{displaymath}

  \end{claim}

  Now for any $k \in [\alpha]$ there are $k$ distinct indices
  $i_1,\dots,i_k \in [\alpha]$ such that 
  \begin{eqnarray}
    \label{eqn:jumpingk}
    \Pr_{L_k \in \mathcal{L}_k}\left[ \Pr_{u\in O_{L_k}}\left[
        f_{i_1}(u)=\cdots=f_{i_k}(u)=1\right] \geq
      \frac{1}{2^{k^2(t+2)}}\right] & \geq & \frac{1}{2^{k(t+2)}}   
  \end{eqnarray}

  The proof is by induction on $k$. Base case is when $k=1$ and
  follows from Equation 5.4 of \cite{amano2005superpolynomial} which
  is established to be true in our setting also. Suppose the claim
  holds for $k \leq l$ and let $k=l+1$. From induction hypothesis we
  get that
  \begin{eqnarray}
    \Pr_{L_l \in \mathcal{L}_l}\left[ \Pr_{u\in O_{L_l}}\left[
        f_{i_1}(u)=\cdots=f_{i_l}(u)=1\right] \geq
      \frac{1}{2^{l^2(t+2)}}\right] & \geq & \frac{1}{2^{l(t+2)}}   
  \end{eqnarray}
  Like in \cite{amano2005superpolynomial} put $c_1=2^{l^2(t+2)},
  c_2=2^{l(t+2)}$  and $c_3 = \alpha$. Note that the bounds $4c_1c_2c_3
  \leq 2^{(l+1)^2(t+2)}$, $2c_1c_1 \leq 2^{(l+1)(t+2)}$ and
  $c_1c_2c_3 \leq 2^{2^{3t}}/8$ are valid in our setting also as
  they do not depend on values of these parameters. 
  Since $t \leq 1/6 \log \log m$, $2^{3t} \leq (\log m)^{1/3}$
  and $2^{2^{3t}} \leq 2^{(\log m)^{1/3}}$ whereas $s^{1/3}$ is
  $m^{\frac{1}{18 \alpha}}\geq 2^{(\log m)(\frac{1}{18 (\log
      m)^{1/6}})}=2^{(\log m)^{5/6}/18} >  2^{(\log
    m)^{1/3}}$. Hence $s^{1/3}/8 \geq 2^{2^{3t}}/8$. 
  . Thus Claim
  \ref{claim:badgraphnextlayer} applies giving us

  \begin{eqnarray}
    \Pr_{L_{l+1} \in \mathcal{L}_{l+1}}\left[ \Pr_{u\in O_{L_{l+1}}}\left[
        f_{i_1}(u)=\cdots=f_{i_{l+1}}(u)=1\right] \geq
      \frac{1}{2^{(l+1)^2(t+2)}}\right] & \geq & \frac{1}{2^{(l+1)(t+2)}}   
  \end{eqnarray}

  The proof of the main theorem is completed by noting that
  $\mathcal{L_\alpha}=\set{V}$ and setting $k$ in
  Equation~(\ref{eqn:jumpingk}) to $\alpha$ gives $Pr_{u\in
    O_V}[\forall i \in [\alpha], f_i(u)=1] > 0$. Thus there exists a
  bad graph $u$ belonging to $\Clique(m,s)^{-1}(0)$ on which all of
  $f_1,\dots,f_\alpha$ outputs $1$, and hence $(u,u^{+})$, where $u^{+}
  \in \Clique(m,s)^{-1}(1)$ is a graph obtained from $u$ by adding an
  edge, which is in $G(f)$ is not covered by any of the $G(f_i)$'s. A
  contradiction. Hence the proof.
  \qed
\end{proof}

Since for a bounded fan-in circuit size lower bound of
$2^{m^{\frac{1}{81 \alpha}}}$ implies a depth lower bound of
$m^{\frac{1}{81 \alpha}}$ we have,

\begin{lemma-a}{~\ref{lemma:depthnegtradeoff}}
  For any circuit $C$ computing $\Clique(m,m^{\frac{1}{6\alpha}})$
  with $\ell$ negations where $\ell \leq 1/6 \log \log m$ , where
  $\alpha=2^{\ell+1}-1$
  \begin{eqnarray*}
    \depth_{\ell}(f) & \geq & m^{\frac{1}{81 \alpha}}
  \end{eqnarray*}
\end{lemma-a}


\end{document}